\documentclass[letterpaper,11pt]{amsart}
\usepackage{amscd, amssymb}
\usepackage{amsmath,amscd}
\usepackage[driverfallback=hypertex]{hyperref}
\usepackage{graphics}
\usepackage{graphicx}
\usepackage{color}
\usepackage{bm}
\usepackage{enumitem}
\setlist[enumerate,1]{label={(\alph*)}}
\setlist[enumerate,2]{label={(\roman*)}}
\input xy
\xyoption{all}

\newtheorem{thm}{Theorem}[section]
\newtheorem{prop}[thm]{Proposition}

\newtheorem{lemma}[thm]{Lemma}
\newtheorem{cor}[thm]{Corollary}

\theoremstyle{definition}
\newtheorem{definition}[thm]{Definition}
\newtheorem{nn}[thm]{Notation}

\theoremstyle{remark}
\newtheorem{rmk}[thm]{Remark}

\newtheorem{obs}[thm]{Observation}

\newcommand{\ignore}[1]{}
\newcommand{\R}{\mathbb R}

\newcommand{\N}{\mathbb N}
\newcommand{\Prob}{\mathbb P}

\newcommand{\Z}{\mathbb Z}

\newcommand{\E}{{\mathbb{E}}}

\usepackage{color}

\definecolor{Red}{rgb}{1,0,0}
\definecolor{Blue}{rgb}{0,0,1}
\definecolor{Olive}{rgb}{0.41,0.55,0.13}
\definecolor{Yarok}{rgb}{0,0.5,0}
\definecolor{Green}{rgb}{0,1,0}
\definecolor{MGreen}{rgb}{0,0.8,0}
\definecolor{DGreen}{rgb}{0,0.55,0}
\definecolor{Yellow}{rgb}{1,1,0}
\definecolor{Cyan}{rgb}{0,1,1}
\definecolor{Magenta}{rgb}{1,0,1}
\definecolor{Orange}{rgb}{1,.5,0}
\definecolor{Violet}{rgb}{.5,0,.5}
\definecolor{Purple}{rgb}{.75,0,.25}
\definecolor{Brown}{rgb}{.75,.5,.25}
\definecolor{Grey}{rgb}{.5,.5,.5}

\title{No percolation in Low temperature Spin Glass}
\author
{Noam Berger\,$^{1}$ \and Ran J. Tessler\,$^2$ }

\begin{document}

\maketitle
\vspace{-5mm}
\centerline{\tiny{$^1$HUJI and TU Munich}}
\centerline{\tiny{$^2$Institute for Theoretical Studies, ETH Z\"urich}}

\begin{abstract}
We consider the Edwards-Anderson Ising Spin Glass model for temperatures $T\geq 0.$ We define the natural notion of Boltzmann-Gibbs measure for the Edwards-Anderson spin glass at a given temperature, and of unsatisfied edges. We prove that for low enough temperatures, in almost every spin configuration 
the graph formed by the unsatisfied edges is made of finite connected components. That is, the unsatisfied edges do not percolate.
\end{abstract}

\pagestyle{plain}

\tableofcontents
\section{Introduction and Statement of Main Result}
\subsection{Some general notations and definitions}
Consider $\Z^2$ as the graph whose vertices are the elements of planar square lattice $\Z^2,$ and there is an edge between any two vertices of $l_1$ distance $1.$ Write
\[
{(\Z^2)^*} = (\frac{1}{2},\frac{1}{2})+\Z^2 = \{(m+\frac{1}{2},n+\frac{1}{2})|~m,n\in\Z^2\}.
\]
There is a canonical identification of the vertices (edges, faces) of ${(\Z^2)^*}$ with the faces (edges, vertices) of $\Z^2.$ Indeed, identify a vertex of ${(\Z^2)^*}$ with the face of $\Z^2$ containing it, identify a face of ${(\Z^2)^*}$ with the unique vertex of $\Z^2$ contained in it, and identify an edge of ${(\Z^2)^*}$ with the unique edge of $\Z^2$ intersecting it.
Note that ${(\Z^2)^*}$ and $\Z^2$ are non canonically isomorphic as graphs with a $\Z^2$ action, for example by the map $v\to(\frac{1}{2},\frac{1}{2})+v.$ For this reason we shall sometimes state claims in terms of $\Z^2,$ but will apply them to $(\Z^*)^2.$ 

Write $[n] = \{0,1,2,\ldots,n-1\},$ and define $[n]^2$ to be the subgraph of the lattice $\Z^2$ over the vertices \[\{(i,j)|i,j\in[n]\}.\]
By abuse of notations we sometimes use $[n]^2$ as the subgraph of $(\Z^2)^*$ over the vertices \[\{(i+(\frac{1}{2},\frac{1}{2}),j+(\frac{1}{2},\frac{1}{2}))|i,j\in[n]\}.\]%

For a graph $G$ write $V=V(G)$ for its set of vertices, and $E=E(G)$ for its set of edges.
Denote by $|G|$ the number of vertices of $G,$ and call it the \emph{size} of $G.$
A \emph{cycle} is a closed path. A path or a cycle is \emph{simple} if it has no self intersections.
~For a subgraph $G\subseteq\Z^2,$ let $G^c$ be the complement subgraph whose vertices are $\Z^2\setminus V(G).$ Let $\partial G$ be the subgraph of $G$ whose vertex set is made of vertices which have a neighbor in $G^c.$ Write 
$\bar{G}=G\cup\partial G^c.$ If $G_1,G_2$ are two subgraphs of $\Z^2,$ write $E(G_1,G_2)$ for the set of edges connecting $G_1$ to $G_2.$ Analogous definitions can be given for subgraphs of $(\Z^*)^2.$
~For a loop $\gamma\subseteq(\Z^*)^2,$ denote by $D_\gamma$ the set vertices which belong to the finite domain bounded by $\gamma.$

\subsection{Edwards-Anderson Spin Glass distributions on $\Z^2$}

An assignment $w: E(\Z^2) \mapsto \R$ of real values to the edges of $\Z^2$ will be called \emph{interactions}. That is, $w_e$ is the interaction along the edge $e \in E(\Z^2).$ A subgraph of $\Z^2,$ together with interactions along its edges is called a \emph{weighted graph}.

A \emph{spin configuration} or \emph{spins} for short, is the assignment $\sigma: \Z^2 \mapsto \{-1, +1\}$ of spin values $\pm 1$ to the vertices of $\Z^2$. If $C$ is a subgraph of $\Z^2$, we shall write $w_C$ and $\sigma_C$ for $w|_{E(C)}$ and $\sigma|_{V(C)}$ respectively. We shall write $\Omega^C$ for  $\{+1, -1\}^{V(C)}$, the set of all spin-configurations restricted to $C$. If, in addition, \emph{boundary conditions} $\tau \in \Omega^{\partial C^c}$ are specified, then we set $\Omega^{C,\tau} := \{\sigma\in\Omega^{\bar{C}}~s.t.~\sigma_{\partial C^c} = \tau\}$.

Given interactions $w$, a finite subgraph $C$ with boundary conditions $\tau \in \Omega^{\partial C^c}$ and an \emph{inverse temperature} parameter $\beta \geq 0,$ we define the Boltzmann-Gibbs distribution $\Prob^{C,\tau}_{w, \beta}$ on $\Omega^{C,\tau}$ by
\[
\Prob^{C,\tau}_{w, \beta} (\sigma) := \frac{1}{Z^{C,\tau}_{w, \beta}} \exp \big\{ -\beta \mathcal{H}_w^{C,\tau}(\sigma) \big\} \, \text{ for }\,\, \sigma \in \Omega^{C, \tau}  \,,
\]
where $\mathcal{H}_w^{C,\tau}$ is the \emph{restricted Hamiltonian}:
\[
\mathcal{H}_w^{C,\tau}(\sigma) := -\sum_{x,y\in V(\bar{C})}w_{xy}\sigma_x\sigma_y,
\]
and the summation is taken over pairs of neighboring vertices. Above $Z^{C,\tau}_{w, \beta}$ is the unique constant which makes $\Prob^{C,\tau}_{w, \beta}$ a probability measure. When $C,\beta$ are understood from the context, we shall omit them from the notation, and write $\Prob^\tau_w.$

We are now ready the define
\begin{definition}\label{def:Gibbs_Boltz}
%
%
%
%
An \emph{Edwards-Anderson spin glass distribution on $\Z^2$} at inverse temperature $\beta \geq 0,$ or EA$_\beta$ spin glass distribution for shortness, is a joint distribution of interactions $w$ and spin configuration $\sigma$ on $\Z^2$, defined as random variables on the same probability space such that
\begin{enumerate}
\item
$(w_e)_{e\in E(\Z^2)}$ are i.i.d. standard Normal random variables.
\item
For every finite subgraph $C$ of $\Z^2,$ 
\begin{equation}
\label{eq:Gibbs_Boltz}
\Prob(\sigma_{\bar{C}} \in \cdot |\sigma_{C^c},w) = \Prob^{C, \sigma_{\partial C^c}}_{w, \beta}(\cdot)  \,.
\end{equation}
\end{enumerate}
\end{definition}

Existence of such measures on $\Z^2$ and, in particular, of such measures which are also invariant under $\Z^2$-translations, follow from the existence of \emph{metastates}, first defined in \cite{AW},\cite{NS3},\cite{NS4}.
A standard construction is as follows. For $n=1,2, \dots$ let $T_n$ be the discrete centered torus in $\Z^2$ of side length $2n+1,$ obtained by identifying vertices at opposite sides of the box $C_n := [-n,n]^2 \cap \Z^2$, thought of as a subgraph of $\Z^2$. Let $\Prob_n$ be any joint distribution of random interactions $w$ and spins $\sigma$ on $\Z^2$ such that (a) and (b) of Definition~\ref{def:Gibbs_Boltz} hold, with $\Z^2$ replaced by $T_n$ in (b). Now, treating $\Prob_n$ as a probability measure on the compact space $\bar{\R}^{E(\Z^2)} \times \Omega^{\Z^2},$ where $\bar{\R}$ is the one-point compactification of $\R,$ the sequence $(\Prob_n)_{n \geq 1}$ is tight and therefore admits some sub-sequential limit $\Prob.$
The limiting distribution $\Prob$ is translation invariant, supported on $\R^{E(\Z^2)} \times \Omega^{\Z^2},$ and satisfies Definition~\ref{def:Gibbs_Boltz}.

\begin{nn}
Write $\mu$ for the i.i.d product measure of $\mathcal{N}(0,1)$ on $\R^{E(\Z^2)}.$
\end{nn}

The limit $\beta\to\infty,$ or equivalently the \emph{zero-temperature} limit is also a subject of interest. In this case the role of EA spin glass distributions is replaced by the notion of \emph{ground states}.
\begin{definition}
A \emph{ground state} is a joint distribution $\phi$ of spins and interactions such that the marginal distribution of the interaction is $\mu,$ and $\phi-$almost every pair $(\sigma,w)$ satisfies the following condition: For any finite $C\subset\Z^2,$ given $\tau = \sigma_{\partial C^c},~\sigma_C$ is the unique minimizer of
$\mathcal{H}^{C,\tau}_w (-).$
\end{definition}
Intuitively this condition means that no finite change may decrease the Hamiltonian. A similar proof shows the existence of ground states. See \cite{NS2} for interesting and detailed discussions.

\subsection{The main result}
\begin{definition}
Given an edge $e=\{i,j\}$ and given a configuration of spins and interactions, whenever
\[
w_e\sigma_i\sigma_j < 0
\]
we say that $e,e^*$ are \emph{unsatisfied}.

An \emph{unsatisfied cycle} is a dual cycle all of whose edges are unsatisfied.
\end{definition}
The main theorem of this paper is
\begin{thm}\label{thm:main_thm}
There exists some finite inverse temperature $\beta^*$ such that for any $\infty\geq\beta >\beta^*$ the following statement holds.
Let $\nu$ be a translation invariant EA$_\beta$ spin glass distribution. Then $\nu-$almost surely every connected component of the cluster of unsatisfied dual edges is finite.
\end{thm}

\begin{rmk}
The \emph{density} of vertices which belong to a random set $\mathcal{P}$ is
\[
\lim_{N\to\infty}\frac{|\{v\in V([N]^2),~v\in\mathcal{P}\}|}{|[N]^2|},
\]
assuming the limit exists almost surely.
~Density of edges, dual vertices or dual edges can be similarly defined.

For any $\beta,$ and any translation invariant ergodic EA$_\beta$ spin glass distribution $\nu,$ the cluster of unsatisfied dual edges has a well defined positive density of vertices.

The fact that the density of vertices is well defined is an immediate consequence of ergodicity. 
~The positivity is a result of the following local argument. Consider a unit lattice square and denote its edges by $\{f_i\}_{i=1}^4.$ With a positive probability the product \[\prod_{i=1}^4 w_{f_i}\] of the interactions of its edges is negative. In this case, for any choice of spins, an odd, hence nonzero, number of dual edges from the vertex dual to that unit square must be unsatisfied. Thus, the density of unsatisfied dual edges in any translation invariant measure is positive.
\end{rmk}
\subsection{Plan of the proof}
The steps of the proof are as follows. In Subsection \ref{sec:1} we establish some preliminary results about random subgraphs of the weighted lattice $\Z^2.$ We then show, in Subsection \ref{sec:2}, that for low enough temperatures there are almost surely only finitely many unsatisfied cycles through any vertex.
In Subsection \ref{sec:3} we prove that from a translation invariant distribution of graphs with only finitely many cycles through each vertex, one can extract, in a translation invariant way, a distribution of spanning forests of these graphs. An analysis of translation invariant distributions of forests and weighted forests with infinite components appears in Subsection \ref{sec:4}. We use this analysis to deduce that such forests cannot appear in the support of translation-invariant EA spin glass measures, and prove the theorem. This is the content of Subsection \ref{sec:5}.

\subsection{Related works and earlier results}
The understanding of spin glass models in large or infinite graphs has been the subject of many studies in physics, mathematics and neuroscience. In \cite{NS1} the notion of ground state have been introduced. 
Questions regarding the multiplicity of ground states in finite dimensional short-range systems, such as the EA Ising spin glass, and in particular the 2D case were the subjects of many studies and simulations (e.g. \cite{EA},\cite{NS1},\cite{NS2},\cite{AD} \cite{ADNS},\cite{BM},\cite{M},\cite{FH},\cite{PY}, \cite{NS5},\cite{Ha}).
The notion of unsatisfied dual edges appeared in \cite{NS1},\cite{ADNS} and was used to define \emph{domain walls}. These were the main tool for investigating metastates in these papers.

The geometry of ground states was studied in \cite{BT}, the M.A. thesis of R.T under the guidance of N.B.
One of the main results there was
\begin{thm}
For any translation-invariant ground state of the EA spin-glass model on $\Z^2$ the cluster of unsatisfied edges is almost surely a forest, all of whose connected components are finite.
\end{thm}
Theorem \ref{thm:main_thm} extends this result from ground states, which correspond to temperature $0,$ to positive low temperatures.

\subsection{Acknowledgments}
We would like to thank Gady Kozma and Oren G. Louidor for fruitful and interesting discussions.
Large portion of the research leading to this paper was performed in the Hebrew university of Jerusalem. R.T. is supported by Dr. Max R\"ossler, the Walter Haefner Foundation and the ETH Z\"urich
Foundation.

\section{No Percolation for Low temperatures}
\begin{nn}
Given a probability measure $\kappa,$ we denote by $\Prob_\kappa,~\E_\kappa$ the probability and expectation with respect to $\kappa,$ respectively.
\end{nn}
Throughout this article $\log{n}$ will denote logarithm with respect to the natural base $e.$
\subsection{General properties of graphs in the EA spin glass model}\label{sec:1}
\begin{nn}
For a subgraph $G\subseteq\Z^2$ or of $(\Z^2)^*$ and an interaction $w$ write
\[
w(G) = \sum_{e\in G}w(e),~~|w|(G) = \sum_{e\in G}|w|(e).
\]
$w(G)$ is the \emph{weight} of $G.~|w|(G)$ is called the \emph{absolute weight} of $G.$

Let $CG(n)$ be the set of connected subgraphs of $[n]^2.$ Denote by $CG(n,m),$ $CG(n,\geq m)$ the subset of $CG(n)$ whose elements are graphs of size $m,\geq m,$ respectively.
\end{nn}
\begin{obs}
Any finite connected graph $G\subseteq\Z^2$ satisfies
\[|G|-1\leq  |E|\leq 2|G|.\]
\end{obs}
Indeed, the lower bound is achieved only for trees. The upper bound is a consequence of the fact any degree is bounded by $4.$

\begin{lemma}\label{not_2_many_connected}
The number of connected subgraphs $G$ of $\Z^2$ of size $n,$ or with $n$ edges, which contain the origin, is between $2^n$ and $32^n.$
\end{lemma}
This lemma is standard and well-known. Nevertheless,  we give a proof for the sake of self-containedness.
\begin{proof}
We prove for vertices. The proof for edges is similar.
The lower bound can be easily observed from considering only simple paths from the origin to which are either monotonic nondecreasing in any coordinate or monotonic decreasing in any coordinate. The number of paths of each of the two types is $2^{n-1}.$

For the upper bound, let $G$ be a graph as in the statement of the lemma. $G$ contains a spanning tree $T$ of $n$ vertices.
Starting from the origin, there is at least one directed path $P$ which goes through every edge of the tree exactly twice, once in every direction. Thus, any spanning tree of $G$ induces some paths from the origin to itself of length $2n-2$. Any directed path of length $2n-2,$ starting from the origin, may be induced by at most one tree $T$ of size $n$ in the procedure just described. Hence, the number of trees of size $n,$ containing the origin is bounded by the number of directed paths of length $2n-2$ from the origin. This number is exactly
\[4^{2n-2},\]
as each vertex along the path has $4$ possibilities for the next edge.

Now, given a tree $T$ of size $n,$ the number of graphs which contain $T$ as a spanning tree is no more than $2^{n+1}.$ Indeed, any such graph $G$ is obtained from $T$ by adding edges which do not add new vertices. As the number of vertices is $n,$ and each degree in $\Z^2$ is $4,$ the total number of possible edges between any $n$ vertices is no more than $2n.$ But $T$ already uses $n-1,$ and so at most $n+1$ potential edges are left, each may appear or not appear in $G.$

Putting everything together, the number of size $n$ connected graphs containing the origin is no more than
\[
4^{2n-2}\cdot2^{n+1} = 2^{5n-1}<2^{5n}=32^n,
\]
as claimed.
\end{proof}

\begin{lemma}\label{heavy_subgraphs}
There exist $\lambda_1,\lambda_2 > 0$ which satisfy the following. $\mu-$almost surely for every dual vertex $x,$ there exists $N\in\N$ such that for all $n>N,$ and for every connected subgraph $G$ of $x+[n]^2$ of size at least $\log{n}$ it holds that
\[
\lambda_1 |E|\geq |w|(G)\geq \lambda_2 |E|.
\]
\end{lemma}
\begin{proof}
Write $A=32.$ By Lemma \ref{not_2_many_connected}, for any $m\in \N,$ the number of connected subgraphs of $[n]^2$ of size $m$ is no more than $n^2 A^m.$

Let $\{x_n\}_{n\in\N}$ be a sequence of i.i.d. random variables,
\[
x_n\sim \mathcal{N}(0,1).
\]
To establish the upper bound, recall from standard large deviation principles, the existence of a function
\[
I:\R_+\to\R_+,~~I(a)\to\infty,~~a\to \infty,
\]
such that
\[
\Prob_\mu(\sum_{i=1}^k |x_n| > a k)<e^{-I(a)k}.
\]
See \cite{DZ} or any other standard text for details.

For fixed $n,m$ the union bound gives
\[
\Prob_\mu(\bigcup_{G\in CG(n,\geq m)}\{|w|(G)> a |E(G)|\})<n^2\sum_{k\geq m-1}^{n^2} A^k e^{-I(a)k} = n^2\sum_{k\geq m-1}^{n^2}f(a)^k
\]
where $f(a) = \frac{A}{e^{I(a)}}.$ In case $f(a)<1,$
\[
\Prob_\mu(\bigcup_{G\in CG(n,\geq m)}\{|w|(G)> a |E(G)|\})<n^2\sum_{k\geq m-1}^\infty f(a)^k = n^2\frac{f(a)^{m-1}}{1-f(a)}.
\]
Choose $\lambda_1$ large enough so that $f(\lambda_1) <\frac{1}{100}.$
~With this $\lambda_1,$ for all $n$ large enough
\[
\frac{f(\lambda_1)^{\log{n} - 1}}{1-f(\lambda_1)}< 2f(\lambda_1)^{\log{n} - 1} < 2(\frac{1}{100})^{\log{n}-1}< n^{-4}.
\]
By summing over $n$ the following inequality is obtained,
\[
\sum_{n=1}^\infty \Prob_\mu(\bigcup_{G\in CG(n,\geq \log{n})}\{|w|(G)> \lambda_1 |E(G)|\})<M+\sum_{n=1}^\infty\frac{1}{n^2} < \infty,
\]
where $M$ is some positive constant.

A standard Borel-Cantelli argument now implies the existence of $N$ such that for all $n > N,$
and for every $G\in CG(n,\geq \log{n}),$ there holds
\[
|w|(G)\leq \lambda_1 |E(G)|,
\]
as needed.

Regarding the lower bound,
there is a function
\[
J:\R_+\to\R_+,~~J(a)\to \infty,~~a\to 0,
\]
such that
\begin{equation}\label{eq:J}
\Prob_\mu(\sum_{i=1}^k |x_n| < a k)<e^{-J(a)k}.
\end{equation}
Indeed, if $\sum_{i=1}^k |x_n| < a k,$ then at least $\frac{k}{2}$ of the variables $|x_i|$ are smaller than $2a.$
Clearly,
\begin{align*}
\Prob_\mu(\exists S\subseteq\{1,\ldots,k\},~\text{such that}~|S|=\frac{k}{2}~\text{and}~& |x_i|<2a,~i\in S)<\\&<\binom{k}{k/2}p_a^{k/2}
<2^k p_a^{k/2}
\end{align*}
where $p_a$ is the probability that $|x_a|\leq 2a,\lim_{a\to 0}p_a=0.$ The function $J(a)=-\log{(2\sqrt{p_a})}$ satisfies equation \ref{eq:J}. The proof continues from here exactly as in the upper bound.
\end{proof}
Following the steps of the proof, the next corollary is obtained
\begin{cor}\label{rmk:heavy_sub}
The constants $\lambda_1,\lambda_2>0$ can be chosen in a way that
\[\Prob_\mu(\exists G\in CG(n,\geq \log n)~\text{s.t.}~|w|(G)> \lambda_1 |E(G)|~\text{or}~|w|(G)< \lambda_2 |E(G)|)\]
is less than $n^{-4}$ for all $n\geq 1.$
\end{cor}
\subsection{Geometry of configurations in EA spin glass distributions}\label{sec:2}
The following observation plays a key role in this paper.
\begin{obs}\label{obs:flip}
Let $\nu$ be a translation invariant EA$_\beta$ spin glass distribution, and let $\gamma$ be a finite dual graph which is a union of disjoint simple dual cycles.
Then for any $c>0,$
\begin{equation}\label{eq:general_flip}
\Prob_\nu(w(\gamma)\leq-c) < e^{-2\beta c}.\
\end{equation}
In addition, for any subgraph $D,$ with $D_\gamma\subseteq D,$ and boundary conditions $\tau,$
\begin{equation}
\label{eq:flip}
\Prob_w^{D,\tau}(\gamma~\text{is unsatisfied}) < e^{-2\beta |w|(\gamma)}.
\end{equation}
\end{obs}
\begin{proof}
Denote by $A_{\gamma}$ the event that $w(\gamma)<-c.$
Pair each spin configuration $\omega$ with $w(\gamma)<-c$ with the configuration $\bar\omega$ obtained from $\omega$ by flipping all the spins in $D_\gamma,$ the domain bounded by $\gamma.$ Note that in $\bar\omega,~w(\gamma)>c.$
This pairing is one-to-one on its image, which is disjoint from $A_\gamma.$ Now, for any domain $D\subseteq\Z^2,$ with $D_\gamma\subseteq D,$ given any specific boundary conditions $\tau$ for $D,$ and any configuration $\omega\in\Omega^{D, \tau},$ if $w(\gamma)<0,$ then
\[
-\mathcal{H}_w^\tau(\omega) = -\mathcal{H}_w^\tau(\bar\omega) - 2\beta |w(\gamma)|.
\]
Hence
\[
\Prob_w^\tau(\omega) = \frac{\Prob_w^\tau(\bar\omega)}{e^{2\beta |w(\gamma)|}}.
\]
Denote by $A_{\gamma,D,\tau}\subseteq \Omega^{D, \tau}$ the set of spin configurations with $w(\gamma)<-c.$
For $\omega\in A_{\gamma,D,\tau},~\Prob_w^\tau(\omega)\leq \frac{\Prob_w^\tau(\bar\omega)}{e^{2\beta c}},$ therefore
\begin{align*}
\Prob_w^\tau(A_{\gamma,D,\tau}) = \sum_{\omega\in A_{\gamma,D,\tau}}\Prob_w^\tau(\omega)=
\frac{\sum_{\omega\in A_{\gamma,D,\tau}}e^{-\beta\mathcal{H}_w^\tau(\omega)}}{\sum_{\omega\in\Omega^{D,\tau}}e^{-\beta\mathcal{H}_w^\tau(\omega)}}
\\ \leq \frac{\sum_{\omega\in A_{\gamma,D,\tau}}e^{-\beta\mathcal{H}_w^\tau(\omega)}}{\sum_{\omega\in A_{\gamma,D,\tau}}e^{-\beta\mathcal{H}_w^\tau(\omega)}+e^{-\beta\mathcal{H}_w^\tau(\bar\omega)}}
\\ \leq
\frac{1}{1+e^{2\beta c}}.
\end{align*}

The second claim follows from the same argument by taking $c$ to be $|w|(\gamma).$
\end{proof}

\begin{lemma}\label{no_long_unsat_cycles}
There exists a positive $\beta_0$ such that for any inverse-temperature $\beta >\beta_0,$ and every translation invariant EA$_\beta$ spin glass measure $\nu,~\nu-$almost-surely there is only a finite number of unsatisfied cycles through any dual vertex.
\end{lemma}
\begin{proof}
Choose $\lambda_2$ as in Lemma \ref{heavy_subgraphs}, and fix a dual vertex $x$. Put $A=32,$ the base of the exponential growth of Lemma \ref{not_2_many_connected}. Let $\beta$ be an arbitrary inverse temperature. Given a dual cycle $\gamma,$ denote by $A_{\gamma}$ the event that $\gamma$ is unsatisfied. 
Let $B_i$
be the event that there exists a dual unsatisfied cycle through $x$ of length exactly $i.$
The lemma is equivalent to proving
\begin{equation}\label{eq:bc}
\Prob_\nu(\bigcap_{j\geq 0} \bigcup_{j>i}B_{i}) = 0.
\end{equation}
Lemma \ref{heavy_subgraphs} tells us that for $i$ large enough, every cycle $\gamma$ through $x$ of length $i$ satisfies
\begin{equation}\label{eq:qtsc}
|w|(\gamma)\geq \lambda_2 i.
\end{equation}
Let $\iota$ be the smallest number such that for all $i\geq\iota$, every cycle $\gamma$ through $x$ of length $i$ satisfies \eqref{eq:qtsc}, $\Prob_\nu(\iota<\infty)=1$.

Using equation \ref{eq:flip}, for a cycle $\gamma$ of length $i$, for every $n<i,$
\[
\Prob_\nu(\gamma\mbox{ is unsatisfied}|\iota=n)\leq e^{-2\beta\lambda_2 i}.
\]
There are no more than $A^i$ possible cycles of length $i$ through $x,$ by Lemma \ref{not_2_many_connected}. Thus, by the union bound,
\[
\Prob_\nu(B_i|\iota=n)\leq e^{-2\beta\lambda_2 i}A^i\leq (Ae^{-2\beta\lambda_2})^i.
\]
Choose $\beta_0$ so that
\[
Ae^{-2\beta_0\lambda_2} < 1,
\]
and hence for any $\beta>\beta_0,~
~Ae^{-2\beta\lambda_2} < 1.$

Thus, for all $\beta>\beta_0$ and every $n$,
\[
\sum_{i=1}^\infty\Prob_\nu(B_i|\iota=n)<\infty
\]
and by Borel-Cantelli
\[
\Prob_\nu(\bigcap_{j\geq 0} \bigcup_{j>i}B_{i}|\iota=n)=0.
\]

\eqref{eq:bc} follows.

\end{proof}
%

From now on we fix a $\beta$ which guarantees that, $\nu-$almost surely, for any dual vertex $x,$ there is only a finite number of unsatisfied cycles passing through $x.$
\subsection{Extracting a forest}\label{sec:3}
For the following lemma, consider a measure $\rho$ of subgraphs of $\Z^2$, invariant under the translations group action, or a joint distribution of subgraphs of $\Z^2$ and interactions, invariant under the diagonal action of the translations group.
Assume, moreover, that for almost every subgraph $G$ in the support of $\rho,$ and any vertex $v\in G$ there are only finitely many simple cycles in $G$ containing $v.$
\begin{lemma}\label{lem:extract forest}
Under the above assumptions there exists a translation invariant process which extracts a spanning forest for $G,$ whose connected components are exactly the connected components of $G.$
\end{lemma}

For the proof we first define the following \emph{loop erasing process}. For each cycle in $G$ attach a Poisson clock with rate which depends on the cycle. When the clock rings, if the cycle still exists, uniformly at random delete one edge from it. In case the cycle does not exist anymore, nothing happens. Call the resulting graph $G_\infty.$

In more formal terms, let $\Gamma$ be the collection of all simple cycles in $\Z^2.$
To each simple cycle $\gamma$ attach a Poisson clock of rate $r_\gamma,$ to be determined soon, and a uniformly
chosen edge $e_\gamma$ of the cycle.

Note that $\Z^2$ acts on $\Gamma$ by translations.
For a given edge $e,$ let $\Gamma_e$ be the set of all simple cycles $\gamma$ which contain it.
For $\gamma\in\Gamma$ let $l_\gamma$ be the length of $\gamma.$
Choose the rates $\{r_\gamma\}_{\gamma\in\Gamma}$ of the clocks in a translation invariant manner, under the requirement that
\begin{equation}
\label{eqn:WellDefinedness}
\sum_{\gamma\in \Gamma_e}r_{\gamma} l_{\gamma} < \infty,
\end{equation}
for some, hence every, edge $e,$ holds.
This requirement can be easily fulfilled, for example by taking $r_\gamma$ of the form $\exp\{-al_\gamma\},$ for some large enough $a.$

Fix $\theta > 0,$ small enough so that
\begin{equation}\label{eq:theta}
\theta\sum_{\gamma\in \Gamma_e}r_{\gamma} l_{\gamma} <1.
\end{equation}

Define a sequence of decreasing graphs $G_0\supseteq G_1\supseteq G_2\supseteq\ldots$ as follows. $G_0=G.$ Suppose $G_n$ has been defined, write $\Gamma_n$ for the set of simple cycles contained in it. For an edge $e\in G_n$ set $D^e_0=\emptyset.$ Let $D^e_1\subseteq\Gamma_n$ be the set of simple cycles which contain $e$ and whose clock rings during the time interval $[n\theta,(n+1)\theta].$ Note that almost surely no clock rings at any time of the form $m\theta,m\in\mathbb{N}.$ A simple cycle $\gamma$ belongs to the random subset $D^e_m\subseteq\Gamma_n$ if $\gamma$ intersects some $\gamma'\in D^e_{m-1}$ and
 the clock of $\gamma$
 rings in the time interval $[n\theta,(n+1)\theta].$ Let $E^e_m$ be the set of edges which belongs to cycles in $D^e_m.$ The sequence $\{|E^e_m|\}_{m\in\mathbb{N}}$ is stochastically dominated by the growth of a Galton-Watson tree, with expected number of offsprings being $\theta\sum_{\gamma\in \Gamma_e}r_{\gamma} l_{\gamma}.$ As this number is smaller than $1,$ the tree is subcritical, and the processes $\{E^e_m\},\{D^e_m\}$ are almost surely finite.
Write \[E^e=\bigcup_{m\geq 1}E^e_m,~~D^e=\bigcup_{m\geq 1}D^e_m.\]
Note that $E^e=E^f,D^e=D^f$ for any $f\in E^e.$

Almost surely no two clocks ring at the same time, thus we can order the cycles in $D^e$ according to the times their clocks ring, $\gamma_1,\ldots,\gamma_r,$ where the clock of $\gamma_i$ rings at time $t_i\in[n\theta,(n+1)\theta],$ and $t_i>t_{i-1}.$

Now, define the sets $H^e_i\subseteq E^e,$ for $i=1,\ldots,r$ by $H^e_1=e_{\gamma_1},$ where $e_\gamma$ is the randomly chosen edge of $\gamma,$ as above. $H^e_{i+1}=H^e_{i}$ if $\gamma_{i+1}\cap H^e_i\neq\emptyset.$ Otherwise $H^e_{i+1}=H^e_{i}\cup e_{\gamma_{i+1}}.$
Write $H^e=H_r^e.$

An edge $e\in E(G_n)$ belongs to $E(G_{n+1})$ if and only if $e\notin H^e.$ Define $G_\infty=\bigcap_{n\geq 0} G_n.$
\begin{rmk}
It can be shown, although it is not needed for this paper, that the resulting graph $G_\infty$ is independent of the choice of $\theta,$ as long as \eqref{eq:theta} is satisfied. Moreover, with the same method one can actually construct a random decreasing family of graphs $(G_t)_{t\in[0,\infty)}$ with the property that
for all $t,$ an edge $e$ belongs to $\lim_{s\nearrow t} G_s\setminus G_t$ if and only if the following two requirements hold.
\begin{enumerate}
\item
There exists a cycle $\gamma$ in $\lim_{s\nearrow t} G_s,$ whose clock rang at time $t.$
\item
$e=e_{\gamma}.$
\end{enumerate}
It can then be proved that $\lim_{t\to\infty} G_t=G_\infty.$
\end{rmk}
\begin{prop}\label{prop:no_cycles_extracted}
The distribution of the random graph $G_\infty$ is translation invariant. Almost surely $G_\infty$ contains no cycles.
\end{prop}
\begin{proof}
The first statement is obvious by construction. Regarding the second, note that almost surely the clock of any simple cycle $\gamma$ should ring at least once, hence in at least one time interval $[n\theta,(n+1)\theta].$ By construction, at least one edge of $\gamma$ will not appear in $G_{n+1}.$
\end{proof}

\begin{proof}[Proof of Lemma \ref{lem:extract forest}]
Apply the loop erasing process defined above, with a translation invariant choice of rates which satisfies \eqref{eqn:WellDefinedness}. By Proposition \ref{prop:no_cycles_extracted} the resulting graph $G_\infty$ is almost surely a forest.
In addition, if $u,v$ are connected in $G$, they stay connected in any $G_n,$ for finite $n.$ Indeed, no edge removal in the process changes connectivity properties, as only one edge which lies on a cycle is eliminated at each step. It only remains to prove that $u,v$ are still connected in the infinite time limit.

The only way $u,v$ can be connected in any finite time, but not in the limit, is if for any $n,$ and any path $P$ between $u$ and $v$ in $G_n,$ there exists $m>n$ such that $P$ is not contained in $G_m.$ 

It is enough to show that this scenario cannot happen for neighboring $u,v.$ Indeed, if any neighboring $u,v$ remain in the same connected component of $G_\infty,$ the same will hold for any $u,v$ which are connected in $G.$

Let $P_0$ be the path made of the single edge $e$ between $u, v$ in $G.$ Let $P_1$ be a simple path connecting $u,v$ after $e$ has been removed. $P_2$ be a simple path between $u$ and $v$ which still connects them after an edge of $P_1$ has been removed etc. Consider the simple paths $P_1,P_2,\ldots.$ They are all distinct and do not contain $e.$ Thus, the cycle $P_m\cup\{e\}$ are all simple and distinct.
But this contradicts our assumption that any vertex is contained in a finite number of simple cycles in $G,$ and the lemma follows.
\end{proof}

\subsection{Translation invariant measure of lattice forests}\label{sec:4}
Our next goal is to analyze the structure of lattice forests which belong to the support of a translation invariant measure.
\begin{definition}\label{def:trees}
Let $T$ be an infinite tree of bounded degree. It is \emph{single-infinite} if it does not contain two disjoint one-way infinite simple paths. It is \emph{bi-infinite} if it has two disjoint one-way-infinite paths, but not three. Otherwise it is \emph{multi-infinite}.

In a single-infinite tree, any vertex $v$ has a single one-way infinite path which starts at $v$. If $u,v$ belong to the same single-infinite tree, we say that $u$ is \emph{behind} $v$ if when we erase $v$ from the tree the component which contains $u$ is finite.

In a bi-infinite tree, there is a single two-way infinite path. It is called the \emph{path} of the tree, and is denoted by $P(T)$, or simply $P,$ if the tree is understood from context.

In a multi-infinite tree, a vertex $v$ with the property that $T\setminus \{v\}$ has at least three infinite components is called an \emph{encounter point}.

A forest all of whose components are single-infinite (bi-infinite) trees is called \emph{single infinite (bi-infinite).}

\end{definition}
\begin{rmk}\label{rmk:rmk forests}
It follows from K\"{o}nig's Lemma \cite{K}, that any infinite tree of bounded degree contains at least one-way infinite path. If it has no two such disjoint paths, then unless the tree itself is a one-way infinite path (which is impossible in the translation invariant context), it has no unique or canonical one-way infinite path.

In a multi-infinite tree, an encounter point always exists. Moreover, it is easy to see, following \cite{BK}, that the number of encounter points in some finite domain in the graph is never higher than the maximal number of disjoint one-way infinite paths which intersect the boundary of the domain (although these paths can be decomposed in more than one way, the maximal number of paths is finite and well defined).
%
\end{rmk}

The next beautiful well known lemma is due to Burton and Keane \cite{BK}.
\begin{lemma}\label{lem:3_inf}
Let $\eta$ be a measure of planar lattice forests all of whose connected components are infinite.
Suppose $\eta$ is invariant under the translation groups's action. Then $\eta-$almost always, any tree component of the forest is either single-infinite or bi-infinite.
\end{lemma}
\begin{proof}
Suppose there are some multi infinite components, in that case there are encounter points as well.
Consider a large square $F=[N]^2\subset\Z^2.$ Following the last part of Remark \ref{rmk:rmk forests}, the number of encounter points inside $F$ cannot be larger than the number of boundary edges of $F.$ Write $Encounter(F)$ for the number of encounter points in $F.$ Thus,
\[
\E_\eta(Encounter(F))\leq |E(\partial F)|=4N-4.
\]

On the other hand, if the probability for having encounter points is positive, then
there exists a positive number $\rho$ such that an arbitrary vertex $v$ is an encounter point with probability $\rho.$ The linearity of expectation then yields
\[
\E_\eta(Encounter(F))=\rho |F|.
\]
Thus
\[
| \partial F| \geq \rho |F|=\rho N^2.
\]
But then
\[
\frac{4N-4}{N^2} \geq \rho.
\]
Since this inequality must hold for all $N,~\rho$ must vanish.
\end{proof}

Let $F$ be a lattice forest, and fix $N\in\N.$ A \emph{bridge} is a simple path in $F\cap[N]^2$ whose endpoints belong to $\partial[N]^2.$
\begin{lemma}\label{lem:1_inf_trees}
Let $\eta$ be a translation-invariant measure of single-infinite lattice forests.
Then there exists a constant $C>0$ which satisfies the following. For all $N,$ the $\eta-$expected number of edges (vertices) of the forest which lay on bridges is at least $CN\log{N}$. Moreover, $C$ can be taken to be in the form $C'p_\eta$ where $C'$ is a universal constant and $p_\eta$ is the probability an edge (vertex) belongs to the forest.
\end{lemma}
\begin{proof}
We prove for vertices, the proof for edges is similar. Fix $N,$ denote by $V_N$ the (random) set of vertices which lay on bridges.

We use the following theorem, which is based on the mass transport principle, and was first proven in \cite{BZZ}. We state it using our conventions. 
\begin{thm}
For $n\geq 0,$ and $v\in(\Z^2)^*,$ denote by $A_{v,n}$ the event that
\begin{enumerate}
\item
$v$ belongs to a tree component $T$ of the forest.
\item
There exists $u\in T$ such that $u$ is behind $v$ in $T$ and
\[
|u-v|_\infty = n.
\]
\end{enumerate}
Similarly, denote by $A_v$ the event that $v$ belongs to a tree component $T$ of the forest.
There exists a universal constant $C_0,$ which does not depend on $\eta$ such that for all $n,$ and every vertex $v$ of the lattice
\[
\Prob_\eta(A_{v,n}|A_v)\geq\frac{C_0}{n}.
\]
\end{thm}
We use the notations of the last theorem. Set $p_n = \Prob_\eta(A_{v,n}).$ By translation invariance this probability does not depend on $v.$
Thus $\Prob_\eta(A_v)$ is some well defined constant which does not depend on $v.$ Hence
\[
p_n\geq\frac{D}{n},
\]
where the constant $D$ is $C_0\Prob_\eta(A_v).$

Denote by $A_{v,n}^{L}~(A_{v,n}^{R})$ the event
\begin{enumerate}
\item
$v$ belongs to a tree component $T$ of the forest.
\item
There exists $u\in T$ such that $u$ is behind $v$ in $T$ and \[
|u-v|_\infty = n.
\]
\item
The $x$ coordinate of $u-v$ is $n~(-n).$
\end{enumerate}
Similarly, denote by $A_{v,n}^{U},A_{v,n}^{D},$ the similarly defined events, only with last requirement being that the $y$ coordinate of $u-v$ is $n$ or $-n,$ respectively.
Write $p_n^\alpha$ for the probability of $A_{v,n}^\alpha,~\alpha\in\{L,R,U,D\}.$
By the union bound
\[
p_n^R+p_n^L+p_n^U+p_n^D \geq p_n.
\]

For any $n,$ let $J^R_n\subseteq [N]^2$ denote the following set of vertices:
$v\in J^R_n$ if the boundary vertex $u\in\partial[N]^2$ closest to $v$ in $l_\infty$ norm is unique, located on the right vertical boundary line of the square, and $|u-v|_\infty=|u-v|_1 =n.$ Similarly define $J^L_n, J^U_n, J^D_n.$ The sets $\{J^\alpha_m\}$ for different $\alpha\in\{L,R,U,D\}, 1\leq m <N/2$ are disjoint. Write
\[
J_n = |J_n^R| = |J^L_n|=|J^U_n|=|J^D_n|=N-2n-2.
\]
For $v\in J_n^\alpha,$ if $A^\alpha_{v,n}$ holds, then in particular $v\in V_N.$ Thus,
\[
 p_n^\alpha\leq \Prob_\eta(v\in V_N),
\]
therefore, by linearity of expectation,
\begin{equation}\label{ineq:nlogn}
\sum_{n=1}^{\frac{N}{3}}J_n(p_n^R+p_n^L+p_n^U+p_n^D)\leq \E_\eta(|V_N|).
\end{equation}
%
%
Combining inequality \ref{ineq:nlogn} with the above estimates gives
\[
\E_\eta(|V_N|)\geq\sum_{n=1}^{\frac{N}{3}}J_n p_n\geq\sum_{n=1}^{\frac{N}{3}}J_n\frac{D}{n}
\\\geq\sum_{n=1}^{\frac{N}{3}}|N-2n-2|\frac{D}{n}\geq C(N\log{N}),
\]
for some constant $C>0$ which depends linearly on $\Prob_\eta(A_v).$
\end{proof}

\subsection{Proof of the main theorem}\label{sec:5}
We are now equipped with enough tools to prove the Theorem \ref{thm:main_thm}.
\begin{proof}[Proof of Theorem \ref{thm:main_thm}.]
Let $\beta_0$ an inverse temperature for which the statement of Lemma \ref{no_long_unsat_cycles} holds.
Let $\beta^*\geq \beta_0$ be an inverse temperature to be determined later.
Fix $\beta>\beta^*,$ and suppose to the contrary that with some positive probability there is an infinite component of unsatisfied dual edges.
~Using Lemma \ref{lem:extract forest} we can extract from $\nu$ a translation invariant measure $\eta$ of unsatisfied dual forests which have infinite components with some fixed positive probability. By Lemma \ref{lem:3_inf}, almost surely all the infinite components are either bi-infinite or single-infinite.

In case there is a positive probability for bi-infinite components, let the random set $U$ be their union, and $P=P(U)$ be the union of the paths of the bi-infinite tree components, where the paths are defined in Definition \ref{def:trees}. By translation invariance there should be a positive number $\rho_0$ with
\[\Prob_\eta(e\in P)=\rho_0.\]
Otherwise with probability $1$ the paths are empty, hence $U$ itself must be empty.

Write \[P_N = P\cap[N]^2, ~E_N=|E(P_N)|, ~Y_N=|w|(P_N).\]
Then $\E_\eta E_N=\rho_0N^2,$ and $ E_N\leq N^2$ always, hence \[\Prob_\eta( E_N>\frac{\rho_0}{2}N^2)\geq \frac{\rho_0/2}{1-\rho_0/2}.\] Put $\rho=\frac{\rho_0/2}{1-\rho_0/2}.$

Standard large deviations techniques show the existence of a constant $a>0$ such that
\begin{equation}\label{eq:LDE}
\Prob_\mu(|w|(\partial[N]^2)\geq aN)<e^{-N}.
\end{equation}

We have
%
\begin{align*}
\Prob_\eta(\{ E_N\geq &\frac{\rho_0}{2}N^2\}\cap\{Y_N\geq \frac{2\lambda_2}{3} E_N\}\cap\{|w|(\partial[N]^2)\leq aN)\})\geq \\
&\geq\Prob_\eta( E_N\geq \frac{\rho_0}{2}N^2)-\Prob_\eta(|w|(\partial[N]^2)> aN))-\\
&-\Prob_\eta(\{ E_N\geq \frac{\rho_0}{2}N^2\}\cap\{Y_N< \frac{2\lambda_2}{3} E_N\}\cap\{|w|(\partial[N]^2)\leq aN)\})\geq\\
&\quad\quad\quad\quad\quad\quad\quad\quad\quad\quad\quad\quad\quad\quad\quad\quad\quad\quad\quad\geq\rho-e^{-N}-\Prob_\eta(Q_N),
\end{align*}
where $Q_N$ is the event $\{\exists G\in CG(N,\geq \log N)~|w|(G)< \lambda_2 |E(G)|\}.$
The last inequality holds for $N$ large enough by estimate \ref{eq:LDE}, and the fact that the event $\{Y_N< \frac{2\lambda_2}{3} E_N\}\cap\{ E_N\geq \frac{\rho_0}{2}N^2\}\cap\{|w|(\partial[N]^2)\leq aN)\}$ is contained in $Q_N,$ for large $N.$ Indeed, for large $N,$
\[G=P_N\cup\partial[N]^2\in CG(N,\geq \log N),~~|w|(G)< \lambda_2 |E(G)|.\] Now
\begin{align*}
&\Prob_\eta(\exists G\in CG(N,\geq \log N)~\text{s.t.}~|w|(G)< \lambda_2 |E(G)|)=\\
&=\Prob_\mu(\exists G\in CG(N,\geq \log N)~\text{s.t.}~|w|(G)< \lambda_2 |E(G)|)<N^{-4},
\end{align*}
where the last inequality uses Corollary \ref{rmk:heavy_sub}.

Putting together we get
\begin{align}\label{eq:rho_N}
\Prob_\eta(\{ E_N\geq &\frac{\rho_0}{2}N^2\}\cap\{Y_N\geq \frac{2\lambda_2}{3} E_N\}\cap\{|w|(\partial[N]^2)\leq aN)\})>\\
\notag& \quad\quad\quad\quad\quad\quad\quad\quad\quad\quad\quad\quad\quad\quad\quad>\rho-N^{-4}-e^{-N}.
\end{align}

%

But now note that $P_N,$ if nonempty, divides the square into disjoint regions $\{R_i\}_{i=1}^l,$ that is, the connected components of \[\left([0,N]^2\setminus P_N\right)\cap\Z^2.\]
Each edge of $P_N$ belongs to boundaries of two different regions. The edges of $\partial[N]^2$ separate the regions from $\Z^2\setminus[N]^2,$ each appears as the boundary edge of a single region.

Because different paths in $P$ do not intersect, the regions can be colored in two colors, black and white, so that neighboring regions have different colors.
Let $W$ denote the collection of indices of white regions, and $B$ be the collection of indices of black regions.
In addition, put
\[
v(W)=\sum_{i\in W}\sum_{e\in E(R_i,R_i^c)}v(e),~~v(B)=\sum_{i\in B}\sum_{e\in E(R_i,R_i^c)}v(e).
\]
Since all the edges of $P$ are unsatisfied, we have
\[
v(W)+v(B)=w(\partial [N]^2)-2Y_N.
\]
Thus, we may assume, without loss of generality that
\[
v(W)\leq |w|(\partial [N]^2)/2-Y_N.
\]
This means that flipping the spins in the white regions decreases the Hamiltonian by at least
\[
\Delta H = 2Y_N-|w|(\partial [N]^2).
\]

Using Observation \ref{obs:flip} and the union bound,
\begin{align}
&\Prob_\eta(\{ E_N\geq \frac{\rho_0}{2}N^2\}\cap\{Y_N\geq \frac{2\lambda_2}{3}E_N\}\cap\{|w|(\partial[N]^2)\leq aN)\})=\\
\notag&=\sum_{m=\frac{\rho_0}{2}N^2}^\infty \Prob_\eta(\{ E_N=m\}\cap\{Y_N\geq \frac{2\lambda_2}{3}m\}\cap\{|w|(\partial[N]^2)\leq aN)\})\leq\\
\notag&\quad\quad\quad\quad\quad\quad\leq \sum_{m=\frac{\rho_0}{2}N^2+4N}^\infty A^m e^{-\beta(\frac{2\lambda_2}{3}m-aN)}\leq \sum_{m=\frac{\rho_0}{2}N^2+4N}^\infty A^m e^{-\beta(\frac{\lambda_2}{2}m)}
\end{align}
for $N$ large enough, and $A=32$ is the constant of Lemma \ref{not_2_many_connected}, so that $A^m$ bounds the number of possibilities for the connected graph $P_N\cup\partial[N]^2.$
We now assume that $\beta^*$ satisfies
\begin{equation}\label{eq:beta1stconst}
\text{for all}~\beta >\beta^*,~Ae^{-\beta\lambda_2/2}<\frac{1}{2}.
\end{equation}
For all $\beta>\beta^*$ we get
\begin{align}
&\Prob_\eta(\{ E_N\geq \frac{\rho_0}{2}N^2\}\cap\{Y_N\geq \frac{2\lambda_2}{3}E_N\}\cap\{|w|(\partial[N]^2)\leq aN)\})<\\
\notag& \quad\quad\quad\quad\quad\quad\quad\quad\quad\quad\quad\quad\quad\quad\quad\quad\quad\quad\quad<2(Ae^{-\beta\lambda_2/2})^{\frac{\rho_0}{2}N^2}\to0,
\end{align}
as $N\to\infty.$
This contradicts \eqref{eq:rho_N} whenever $\rho_0\neq 0.$
%
%

It is left only to rule out the existence of single infinite components. Observe that, at least a priori, single infinite trees do not necessarily divide a large square into different regions, even if they intersect it.
Yet, it turns out that under the translation invariance assumption such a partition does exist.
\begin{figure}[h]
\begin{center}
\includegraphics[width=0.7\textwidth]{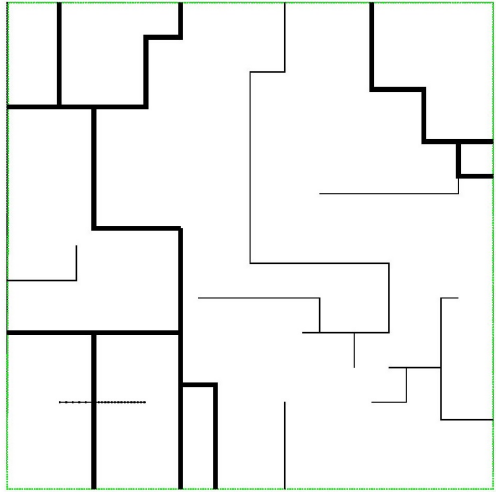}
\caption{The intersection of an infinite lattice tree with a large square. Bold lines represent bridges.}
\end{center}
\end{figure}
Let $\rho_0$ be the probability an edge belongs to an infinite component of the forest. Write $P_N$ for the graph composed of edges in $[N]^2$ which lay on bridges and $E_N=|E(P_N)|.~E_N$ is expected to be at least $cN\log{N},$ for some constant $c$ which \emph{depends linearly} on $\rho_0,$ by Lemma \ref{lem:1_inf_trees}. In addition, $ E_N\leq N^2.$ A computation then shows that \[\Prob_\eta( E_N\geq \frac{c}{2}N\log{N})\geq\frac{\frac{c}{2}N\log{N}}{N^2-\frac{c}{2}N\log{N}}>\frac{c\log{N}}{2N}.\]
Writing $Y_N=|w|(P_N)$ and repeating the argument which led to estimate \ref{eq:rho_N},
we obtain
\begin{align}\label{eq:rho_N2}
&\Prob_\eta( E_N\geq\frac{c}{2}N\log{N},~Y_N\geq d E_N,~\text{and}~|w|(\partial[N]^2)\leq aN))>\\
\notag&\quad\quad\quad\quad\quad\quad\quad\quad\quad\quad\quad\quad\quad\quad\quad\quad\quad>\frac{c\log{N}}{2N} -N^{-4}-e^{-N}.
\end{align}
for $d=\frac{2\lambda_2}{3}.$

$P_N$ divides $[N]^2$ into domains. By using the classical \emph{Five Colors Theorem} (\cite{He}), we can color these domains in five colors so that neighboring domains have different colors. Note that by Euler's formula, the number of different regions inside the square is $E_N-|P_N|<E_N.$
Hence there are at most $5^{E_N}$ ways to color the different domains.

As in the case of bi-infinite trees, there is at least one color, say white, such that the total weight of the boundaries of white domains is at most $-\frac{2}{5}Y_N+\frac{1}{5}|w|\partial[N]^2.$
Thus, flipping the spins inside this region leaves us with
\[
|\Delta H| \geq \frac{4}{5}Y_N-\frac{2}{5}|w|\partial[N]^2,
\]
Again, Observation \ref{obs:flip} and the union bound imply,
\begin{align}
&\Prob_\eta( E_N\geq\frac{c}{2}N\log{N},~Y_N\geq d E_N,~\text{and}~|w|(\partial[N]^2)\leq aN))=\\
\notag&=\sum_{m=\frac{c}{2}N\log{N}}^\infty \Prob_\eta(\{ E_N=m\}\cap\{Y_N\geq dm\}\cap\{|w|(\partial[N]^2)\leq aN)\})\leq\\
\notag&\quad\quad\leq \sum_{m=\frac{c}{2}N\log{N}+4N}^\infty (5A)^me^{-\beta(\frac{4d}{5}m-\frac{2a}{5}N)}\leq \sum_{m=\frac{c}{2}N\log{N}+4N}^\infty (5A)^m e^{-\beta(\frac{dm}{5})}
\end{align}
for $N$ large enough. The $A^m$ term is as before, the $5^m$ term comes from the number of different possible colorings, so together they bound the number of candidates for the set of disjoint cycles we need for applying Observation \ref{obs:flip}.
We now add another constraint on $\beta^*,$
\begin{equation}\label{eq:beta2const}
5Ae^{-\beta^* d/5}<\frac{1}{2},
\end{equation}
For such $\beta>\beta^*$ we have
\begin{align}
&Prob_\eta( E_N\geq\frac{c}{2}N\log{N},~Y_N\geq d E_N,~\text{and}~|w|(\partial[N]^2)\leq aN))<\\
\notag& \quad\quad\quad\quad\quad\quad\quad\quad\quad\quad\quad\quad\quad\quad\quad\quad\quad\quad\quad<2(5Ae^{-\beta d/5})^{\frac{c}{2}N\log{N}}.
\end{align}
This tends to $0$ faster than the right hand side of equation \ref{eq:rho_N2}, unless $\rho_0=0.$
To summarize, fix $\beta^*\geq\beta_0$ which satisfies constraints \eqref{eq:beta1stconst} and \eqref{eq:beta2const}. Then for any $\beta >\beta^*,$ and any translation invariant $EA_\beta$ measure, the probability of existence of infinite unsatisfied clusters is $0.$
In other words, the unsatisfied dual edges do not percolate.
\end{proof}
\section{Open problems}
There are very few rigorous results in the field of EA Ising spin glass model for lattices. In this last section we present some natural open problems.

\emph{Uniqueness of measures.}
Is there a unique EA spin glass measure for any inverse temperature? In case of temperature $0,$ is there a unique (up to global spin flip) ground state?

\emph{Higher dimensions.}
Is there an analog for Theorem \ref{thm:main_thm} for $\Z^d,~d>2?$

\emph{Phase Transition.}
Are there critical temperatures for percolation of unsatisfied dual edges?
In the case of planar square lattice, for $\beta=0,$  one possible translation invariant EA$_0$ measure is the edge bernoulli percolation measure, with $p=0.5.$ Thus, there is no percolation in this case. Yet, Theorem \ref{thm:main_thm} holds for much general infinite planar graphs, and for some of them there is a percolation in $p=0.5,$ as pointed out to us by Gady Kozma. Yet, we do not even know if such a phenomenon is monotonic in temperature. One can try to find critical temperatures and phase transitions for other properties as well.

\emph{Quantitative results}
There are many quantitative questions one can ask. For example, could one calculate the density of unsatisfied edges or their expected value for a given temperature?

\emph{Loop dynamics.}
The zero temperature loop dynamics defined in \cite{BT} is the following dynamical process on spin configurations on weighted graphs. Any finite connected subset $C$ of the graph is attached a Poisson clock of rate $r_C.$ Whenever it rings, if the (restricted) Hamiltonian decreases by flipping all the spins of $C,$ they are flipped. Otherwise they are left unchanged.
It is called the loop dynamics since for a planar graph the connected sets can be represented by their boundary in the dual graph, which is a loop.

In \cite{BT} this process is mainly considered for the lattice $\Z^2.$ It is proved, using techniques similar to those of the loop erasing process, that one can find rates so that the resulting process is well defined and translation invariant. It is moreover proved that all the weak subsequential limits are ground states of the EA Ising spin glass. These results can be easily generalized to much more general graphs.
A natural question is whether or not this process has a limit, and for which families of rates.

There is a natural generalization of the loop dynamics to positive temperatures. Within the above setting, whenever the clock of $C$ rings, if the energy change which occurs when flipping the spins of $C$ is $2\Delta,$ flip $C$ with probability $\frac{e^{-\beta\Delta}}{e^{-\beta\Delta}+e^{\beta\Delta}}.$

Again one can show that there are rates for which this process is well defined and translation invariant. Are the subsequential weak limits of this process EA$_\beta$ spin glass measures? Does it converge to a weak limit?

\bibliographystyle{plain}

\begin{thebibliography}{10}
\bibitem{AD}
L.~P.~Arguin, M.~Darmon.
\newblock On the number of ground states of the Edwards-Anderson spin glass model.
\newblock {\em Ann. Inst. Henri Poincare Prob. Stat.}, 50(1):28--62, 2014.

\bibitem{ADNS}
L.~P.~Arguin, M.~Darmon, C.~M.~Newman, and D.~L.~Stein.
\newblock Uniqueness of ground states for short-range spin glasses in the
  half-plane.
\newblock {\em Comm. Math. Phys.}, 300(3):641--657, 2010.

\bibitem{AW}
M.~Aizenman, J.~Wehr.
\newblock Rounding effects of quenched randomness on first-order phase transitions.
\newblock {\em Comm. Math. Phys.}, 130:489--528, 1990.

\bibitem{BK}
R.~M.~Burton and M.~Keane.
\newblock Density and uniqueness in percolation.
\newblock {\em Comm. Math. Phys.}, 121(3):501--505, 1989.

\bibitem{BM}
A.~J.~ Bray and M.~A.~Moore.
\newblock Chaotic nature of the spin-glass phase.
\newblock {\em Phys. Rev. Lett.}, 58(1):57--60, 1987.

\bibitem{BT}
N.~Berger and R.~J.~Tessler.
\newblock {Geometry and dynamics in zero temperature statistical mechanics
  Models}.
\newblock {\em ArXiv e-prints}, arXiv 1008.5279, 2010.

\bibitem{BZZ}
M.~Bramson, O.~Zeitouni, and M.~P.~W.~Zerner.
\newblock Shortest spanning trees and a counterexample for random walks in
  random environments.
\newblock {\em Ann. Prob.}, 821--856, 2006.

\bibitem{DZ}
A.~Dembo and O.~Zeitouni.
\newblock{\em Large deviations techniques and applications} 38, Springer Science \& Business Media, 2009.

\bibitem{EA}
S.~Edwards and P.~W.~Anderson.
\newblock Theory of spin glasses.
\newblock {\em 1975 J. Phys. F: Met. Phys.} 5, 965, 1975.

\bibitem{FH}
D.~S.~Fisher and D.~A.~Huse.
\newblock Ordered phase of short-range {I}sing spin-glasses.
\newblock {\em Phys. Rev. Lett.}, 56(15):1601--1604, 1986.

\bibitem{Ha}
A.~K.~Hartmann.
\newblock How to evaluate ground-state landscapes of spin glasses
  thermodynamical correctly.
\newblock {\em Phys. Rev. B}, 13:539--545, 2001.

\bibitem{He}
P.~J.~Heawood.
\newblock Map-colour theorem.
\newblock {\em Proc. London Math. Soc.}, 2(1):161--175,
  1949.

\bibitem{K}
D.~K\"onig.
\newblock {\em Theorie der endlichen und unendlichen Graphen: kombinatorische
  Topologie der Streckenkomplexe}.
\newblock Leipzig, 1936.

\bibitem{M}
W.~L.~McMillan.
\newblock Scaling theory of {I}sing spin glasses.
\newblock {\em J. Phys. C: Solid State Physics}, 17(18):3179, 1984.

\bibitem{NS1}
C.~M.~Newman and D.~L.~Stein.
\newblock Are there incongruent ground states in 2{D} {E}dwards-{A}nderson spin
  glasses?
\newblock {\em Comm. Math. Phys.}, 224(1):205--218, 2001.
\newblock Dedicated to Joel L. Lebowitz.

\bibitem{NS2}
C.~M.~Newman and D.~L.~Stein.
\newblock Nature of ground state incongruence in two-dimensional spin glasses.
\newblock {\em Phys. Rev. Lett.}, 84(17):3966, 2000.

\bibitem{NS3}
C.~M.~Newman and D.~L.~Stein.
\newblock Metastate approach to thermodynamic chaos.
\newblock {\em Phys. Rev. E}, 55: 5194--5211, 1997.

\bibitem{NS4}
C.~M.~Newman and D.~L.~Stein.
\newblock Spatial inhomogeneity and thermodynamic chaos.
\newblock {\em Phys. Rev. Lett.}, 76: 4821--4824, 1996.

\bibitem{NS5}
C.~M.~Newman and D.~L.~Stein.
\newblock Multiple states and thermodynamic limits in short-ranged {I}sing
  spin-glass models.
\newblock {\em Phys. Rev. B}, 46(2):973--982, 1992.

\bibitem{PY}
M.~Palassini and A.~P.~Young.
\newblock Triviality of the ground state structure in {I}sing spin glasses.
\newblock {\em Phys. Rev. Lett.}, 83(24):5126--5129, 1999.

\end{thebibliography}

\end{document}